\documentclass[11pt,reqno]{amsart}
\usepackage{amsfonts}
\usepackage{mathrsfs}
\allowdisplaybreaks[4]
\usepackage{graphicx} 
\usepackage{epsfig}
\usepackage{amssymb}
\usepackage{amsmath}
\usepackage{cite}
\newtheorem{theorem}{Theorem}

\newtheorem{proposition}[theorem]{Proposition}

\newtheorem{lemma}[theorem]{Lemma}
\newcommand{\be}{\begin{equation}}
\newcommand{\ee}{\end{equation}}
\newcommand{\bea}{\begin{eqnarray}}
\newcommand{\eea}{\end{eqnarray}}
\newcommand{\ba}{\begin{array}}
\newcommand{\ea}{\end{array}}
\newcommand{\bean}{\begin{eqnarray*}}
\newcommand{\eean}{\end{eqnarray*}}

\newcommand{\pa}{\partial}

\begin{document}
\title{The Additional Symmetries for the BTL and CTL Hierarchies}
\author{Jipeng Cheng$^{1,2}$, Kelei Tian$^2$, Jingsong He$^{1*}$,
 }
\dedicatory {1\ Department of Mathematics, Ningbo University,
Ningbo, Zhejiang 315211, P.\ R.\ China \\
2\ Department of Mathematics, USTC, Hefei,Anhui 230026  , P.\ R.\
China}

\thanks{$^*$Corresponding author. email:hejingsong@nbu.edu.cn; jshe@ustc.edu.cn.}
\begin{abstract}
In this paper, we construct the additional symmetries for the Toda lattice (TL) hierarchies of B type and C type (the BTL and CTL hierarchies), and show their algebraic structures are $w_\infty^B\times w_\infty^B$ and $w_\infty^C\times w_\infty^C$ respectively. And also we discuss the generating functions of the additional symmetries.





\end{abstract}
\maketitle
\section{Introduction}
The Toda lattice (TL) hierarchy was first introduced by K.Ueno and K.Takasaki in \cite{uenotaksasai}
to generalize the Toda lattice equations\cite{toda}.
Along the work of E. Date, M. Jimbo, M. Kashiwara and T. Miwa \cite{DJKM} on the KP hierarchy, K.Ueno
 and K.Takasaki in \cite{uenotaksasai} develop the theory for the TL hierarchy: its algebraic structure, the linearization, the bilinear identity,
 $\tau$ function and so on. Also the analogues of the B and C types for the TL hierarchy, i.e.
  the BTL and CTL hierarchies, are considered in \cite{uenotaksasai}, which are corresponding to
  infinite dimensional Lie algebras $\textmd{o}(\infty)$ and $\textmd{sp}(\infty)$ respectively.
  In this paper, we will focus on the study of the additional symmetries for the BTL and
  CTL hierarchies.

Symmetries have been playing vital roles in the study of the
integrable system. The additional
symmetry\cite{OS86,D93,ASM94,vM94,ASM95,D95,AvM1999,takasaki1996} is
one of the most important symmetries, which has two different
expressions. One of its expressions can be traced back to the master
symmetry\cite{fokas1981,OF1982,F1983,Chen1983,fokas1987,fokas1988a,fokas1988b}.
As an interesting generalization of usual symmetries of partial
differential equations, master symmetries are introduced in
references \cite{fokas1981} and further developed in
references\cite{OF1982,F1983,Chen1983,fokas1987,fokas1988a,fokas1988b}.
The master symmetries are usually for the given explicit soliton
equations, whose remarkable feature is its depending explicitly on
the space $x$ and time $t$ variables for $1+1$-dimensional case.
When considering the integrable hierarchies, the master symmetries
are usually called additional symmetries\cite{OS86}. For the KP
hierarchy, the corresponding additional symmetries are studied in
references \cite{D93,ASM94,vM94,ASM95,D95}, which can be used to
form $w_{\infty}$ algebra when acting on the linear problem, while
the additional symmetry for the TL hierarchy is investigated in
references \cite{ASM95,AvM1999,takasaki1996}, which forms
$w_{\infty}\times w_{\infty}$ algebra when acting on the linear
problem. The other expression for the additional symmetry is in the
form of Sato B\"acklund transformation
  \cite{DJKM} defined by the vertex operator $X(\lambda,\mu)$ acting on the $\tau$ function.
   These two different expressions can be linked by so-called Adler-Shiota-van Moerbeke (ASvM)
    formulas \cite{ASM94,ASM95,vM94,AvM1999,D95} for the continuous integrable systems (the KP
     hierarchy) and also the discrete ones (the TL hierarchy). By the ASvM formulas, the associated
      algebra of additional symmetries can be lifted to its central
extension, the algebra of B\"acklund symmetries. The additional symmetries are involved in so-called
string equation and the generalized Virasoro constraints in matrix models of the 2d quantum gravity
 (see \cite{D93,vM94,takasaki1996} and references therein).

 There are several interesting results about the additional symmetries for the BKP and CKP
 hierarchies from the views of the possible applications related to the string equations. For the
 BKP hierarchy,  the corresponding additional symmetry are constructed by Takasaki
 \cite{takasaki1993} in the operator form, and Virasoro constraints and the ASvM formula have been
 studied by Johan van de  Leur \cite{L95,L96} using an algebraic formalism. Recently, Tu \cite{Tu07}
  gave an alternative proof of the ASvM formula of the BKP hierarchy by using Dickey¡¯s method
  \cite{D95}. And the corresponding string equation was also constructed by Tu in \cite{Tu08}.
As for the CKP hierarchy, the additional symmetry and string equation were well constructed by He
in \cite{he2007}. Inspired by these works and the relation between the TL hierarchy and the
KP hierarchy, we shall in this paper establish the additional symmetries for the BTL and CTL
hierarchies, and then investigate their corresponding algebraic structures, and study
 some interesting properties of them.

This paper is organized in the following way. In Section 2, we recall some basic knowledge about the
BTL and CTL hierarchies. Then, we construct the additional symmetry  for the BTL hierarchy and give
their algebraic structure in Section 3. Next, in Section 4, the additional symmetry for the CTL
hierarchy is also investigated and the corresponding algebraic structure is shown. At last, we
devote Section 5 to some conclusions and discussions.

\section{the BTL and CTL hierarchies}
In this section, we will review some basic facts about the BTL and CTL hierarchies in the style of
 Adler \& van Moerbeke\cite{ASM95,AvM1999}. One can refer to \cite{uenotaksasai} for
 more details about the BTL and CTL hierarchies.

First, consider the algebra
$$ \mathscr{D}=\{(P_1,P_2)\in\textmd{ gl}((\infty))\times \textmd{gl}((\infty))\ | \ (P_1)_{ij}=0 \ \textmd{ for}\  j-i\gg0, \ (P_2)_{ij}=0 \ \textmd{for} \  i-j\gg0\},$$
which has the following splitting:
\begin{eqnarray*}
\mathscr{D}&=&\mathscr{D}_+ +\mathscr{D}_-,\\
\mathscr{D}_+&=&\{(P,P)\in \mathscr{D}\ | \ (P)_{ij}=0 \ \textmd{for} \  |i-j|\gg 0\}=\{(P_1,P_2)\in \mathscr{D}\ | \ P_1=P_2\},\\
\mathscr{D}_-&=&\{(P_1,P_2)\in \mathscr{D}\ | \ (P_1)_{ij}=0\ \textmd{for}\ j\geq i, \ (P_2)_{ij}=0 \ \textmd{for} \  i> j\},
\end{eqnarray*}
with $(P_1,P_2)=(P_1,P_2)_+ +(P_1,P_2)_-$ given by
$$(P_1,P_2)_+=(P_{1u}+P_{2l},P_{1u}+P_{2l}),\ (P_1,P_2)_-=(P_{1l}-P_{2l},P_{2u}-P_{1u}),$$
where for a matrix $P$, $P_u$ and $P_l$ denote the upper (including diagonal) and strictly lower triangular parts of $P$, respectively. For $(P_1,P_2),(Q_1,Q_2)\in  \mathscr{D}$, we define $$(P_1,P_2)(Q_1,Q_2)=(P_1Q_1,P_2Q_2),\quad (P_1,P_2)^{-1}=(P_1^{-1},P_2^{-1}).$$

Then the BTL (or CTL) hierarchy is defined in the Lax forms as
\begin{equation}\label{bctlhierarchy}
    \pa_{x_{2n+1}}L=[(L^{2n+1}_1,0)_+,L]\ \ \ \textmd{and}\ \ \ \pa_{y_{2n+1}}L=[(0,L_2^{2n+1})_+,L],\ \ \ n=0,1,2,\cdots
\end{equation}
where the Lax operator $L$ is given by  a pair of infinite matrices
\begin{equation}\label{laxoperator}
    L=(L_1,L_2)=\Big(\sum _{-\infty<i\leq 1}\textmd{diag}[a_i^{(1)}(s)]\Lambda^i,\sum _{-1\leq i<\infty}\textmd{diag}[a_i^{(2)}(s)]\Lambda^i\Big)\in  \mathscr{D}
\end{equation}
with $\Lambda=(\delta_{j-i,1})_{i,j\in \mathbb{Z}}$,
and $a_1^{(k)}(s)$ and  $a_2^{(k)}(s)$ depending on $x=(x_1,x_3,x_5,\cdots)$ and $y=(y_1,y_3,y_5,\cdots)$, such that
$$a_1^{(1)}(s)=1 \ \ \ and\ \ \ a_2^{(1)}(s)\neq 0\ \ \ \forall s$$
and satisfies the BTL (or CTL) constraint\cite{uenotaksasai}
\begin{equation}\label{bctlconstr}
    L^T=-(J,J)L(J^{-1},J^{-1})\ \left(\textmd{or}\ L^T=-(K,K)L(K^{-1},K^{-1})\right),
\end{equation}
where $J=((-1)^i\delta_{i+j,0})_{i,j\in\mathbb{Z}}$, $K=\Lambda J$ and $T$ refers to the matrix transpose.

The Lax operator of the BTL (or CTL) hierarchy (\ref{bctlhierarchy}) has the representation
\begin{equation}\label{laxwaveexpression}
    L:=W(\Lambda,\Lambda^{-1})W^{-1}=S(\Lambda,\Lambda^{-1})S^{-1}
\end{equation}
in terms of two pairs of wave operators $W=(W_1,W_2)$ and $S=(S_1,S_2)$, where
\begin{eqnarray}
S_1(x,y)=\sum_{i\geq 0} \textmd{diag}[c_i(s;x,y)] \Lambda^{-i},&&S_2(x,y)=\sum_{i\geq 0} \textmd{diag}[c_i'(s;x,y)] \Lambda^{i} \label{smatrices}
\end{eqnarray}
and
\begin{equation}\label{wmatrices}
    W_1(x,y)=S_1(x,y)e^{\xi(x,\Lambda)},\quad W_2(x,y)=S_2(x,y)e^{\xi(y,\Lambda^{-1})}
\end{equation}
with $c_0(s;x,y)=1$ and $c_o'(s;x,y)\neq 0$ for any $s$, and $\xi(x,\Lambda^{\pm1})=\sum_{n\geq0}x_{2n+1}\Lambda^{\pm 2n+1}$. Obviously,
$W=(W_1,W_2)$ are not uniquely determined, but have the arbitrariness
$$W_1(x,y)\mapsto W_1(x,y)f^{1}(\Lambda),\quad W_2(x,y)\mapsto W_2(x,y)f^{2}(\Lambda).$$
Here $f^{1}(\lambda)=\sum_{i\geq0}f_i^1\lambda^{-i}$ and $f^{2}(\lambda)=\sum_{i\geq0}f_i^2\lambda^{i}$ ($f_0^1=1$, $f_0^2\neq0$) are formal Laurent series with constant scalar coefficients. Under an appropriate choice of $f_i(\lambda)$, $W=(W_1,W_2)$ satisfies
\begin{equation}\label{bcwconstraints}
    J^{-1}W_i^TJ=W_i^{-1}\ \textmd{for}\ \textmd{BTL}\ (\ \textmd{or}\ K^{-1}W_i^TK=W_i^{-1}\ \textmd{for}\ \textmd{CTL}), i=1,2.
\end{equation}
The wave operators evolve according to
\begin{eqnarray}
\pa_{x_{2n+1}}S&=&-(L_1^{2n+1},0)_-S,\quad \pa_{y_{2n+1}}S=-(0,L_2^{2n+1})_-S,\label{sevolution}\\
\pa_{x_{2n+1}}W&=&(L_1^{2n+1},0)_+W,\quad \pa_{y_{2n+1}}W=(0,L_2^{2n+1})_+W.\label{wevolution}
\end{eqnarray}

The vector wave functions $\Psi=(\Psi_1,\Psi_2)$ and the adjoint wave functions $\Psi^*=(\Psi_1^*,\Psi_2^*)$, can also be introduced as
\begin{eqnarray}
  \Psi_i(x,y;\lambda)&=& (\Psi_i(n;x,y;\lambda))_{n\in\mathbb{Z}}:=W_i(x,y)\chi(\lambda),\label{wavefunction} \\
  \Psi_i^*(x,y;\lambda)&=&(\Psi_i^*(n;x,y;\lambda))_{n\in\mathbb{Z}}:=(W_i(x,y)^{-1})^T\chi^*(\lambda),\label{adjointwavefunction}
\end{eqnarray}
with $\chi(\lambda)=(\lambda^i)_{i\in\mathbb{Z}}$, and $\chi^*(\lambda)=\chi(\lambda^{-1})$, which satisfy the following differential equations:
\begin{eqnarray}
\pa_{x_{2n+1}}\Psi&=&(L_1^{2n+1},0)_+\Psi,\quad \pa_{y_{2n+1}}\Psi=(0,L_2^{2n+1})_+\Psi,\label{vectorwavefunctionequation}\\
\pa_{x_{2n+1}}\Psi^*&=&-((L_1^{2n+1},0)_+)^T\Psi^*,\quad \pa_{y_{2n+1}}\Psi^*=-((0,L_2^{2n+1})_+)^T\Psi^*.\label{vectoradjointwavefunctionequation}
\end{eqnarray}

At last, in order to define the additional symmetries of the BTL and CTL hierarchies in next
sections, it is necessary to introduce the Orlov-Shulman operator\cite{OS86,ASM95}
\begin{eqnarray}
M&=&(M_1,M_2):=(W_1\varepsilon W_1^{-1},W_2\varepsilon^* W_2^{-1}),\label{osoperator}
\end{eqnarray}
where $\varepsilon=\textmd{diag}(i)_{i\in\mathbb{Z}}\cdot\Lambda^{-1}$ and $\varepsilon^*=-J\varepsilon J^{-1}$. These operators satisfy
\begin{eqnarray}
L\Psi&=&(\lambda,\lambda^{-1})\Psi,\quad M\Psi=(\pa_\lambda,\pa_{\lambda^{-1}})\Psi,\quad [L,M]=(1,1),\label{lmactonwave}
\end{eqnarray}
and
\begin{equation}\label{mevolution}
    \pa_{x_{2n+1}}M=[(L^{2n+1}_1,0)_+,M],\quad \pa_{y_{2n+1}}M=[(0,L_2^{2n+1})_+,M].
\end{equation}

\section{The Additional Symmetry for the BTL Hierarchies}
In this section, we shall construct the additional symmetry for the BTL hierarchy. Just as the case of the TL hierarchy, the Orlov-Shulman operators given by (\ref{osoperator}) still play an important role.

In this case,
we can still similarly introduce the additional independent variables
$x_{m,l}^*$ and $y_{m,l}^*$, and define the additional flows for the BTL hierarchy as
\begin{equation}\label{btladdsym}
    \pa_{x_{m,l}^*}W:=-(A_{1ml}(M_1,L_1),0)_-W,\quad \pa_{y_{m,l}^*}W:=-(0,A_{2ml}(M_2,L_2))_-W,
\end{equation}
where $A_{iml}(M_i,L_i)$ are polynomials in
$L_i$ and $M_i$ and their explicit forms
will be determined next.

The action on ${\Psi}$,
${L}$ and ${M}$ of the additional flows is
given by the following proposition.
\begin{proposition}\label{propaddisymmtodab1}
\begin{eqnarray}
\pa_{x_{m,l}^*}{\Psi}=-(A_{1ml}({M}_1,{L}_1),0)_-{\Psi},&& \pa_{y_{m,l}^*}{\Psi}=-(0,A_{2ml}({M}_2,{L}_2))_-{\Psi},\label{btlwaveaddsymm}\\
\pa_{x_{m,l}^*}{L}=[-(A_{1ml}({M}_1,{L}_1),0)_-,{L}],&& \pa_{y_{m,l}^*}{L}=[-(0,A_{2ml}({M}_2,{L}_2))_-,{L}],\label{btllaxwaveaddsymm}\\
\pa_{x_{m,l}^*}{M}=[-(A_{1ml}({M}_1,{L}_1),0)_-,{M}],&&
\pa_{y_{m,l}^*}{M}=[-(0,A_{2ml}({M}_2,{L}_2))_-,{M}].\label{btloswaveaddsymm}
\end{eqnarray}
\end{proposition}
\begin{proof}
  Firstly, (\ref{btlwaveaddsymm}) follows from
(\ref{wavefunction}) and
(\ref{btladdsym}). Then for the action on ${L}$, according to
(\ref{laxwaveexpression}) and
(\ref{btladdsym}),
\begin{eqnarray*}
\pa_{x_{m,l}^*}{L}&=&
(\pa_{x_{m,l}^*}{W})\Lambda
{W}^{-1}-{W}
\Lambda{W}^{-1}(\pa_{x_{m,l}^*}{W}){W}^{-1}\\
&=&
-(A_{1ml}({M}_1,{L}_1),0)_-{W}\Lambda
{W}^{-1}+{W}
\Lambda{W}^{-1}(A_{1ml}({M}_1,{L}_1),0)_-{W}{W}^{-1}\\
&=&
-(A_{1ml}({M}_1,{L}_1),0)_-{L}+{L}
(A_{1ml}({M}_1,{L}_1),0)_-\\
&=& [-(A_{1ml}({M}_1,{L}_1),0)_-,{L}].
\end{eqnarray*}
And $\pa_{y_{m,l}^*}{L}=[-(0,A_{2ml}({M}_2,{L}_2))_-,{L}]$
can be similarly derived.

At last, by (\ref{osoperator}), a similar proof leads to
(\ref{btloswaveaddsymm}).
\end{proof}
Next, we will show (\ref{btladdsym}) is indeed the symmetry flow of the BTL hierarchy in the
proposition below, and thus it is called the {\it additional symmetry}.

\begin{proposition}\label{propaddisymmtodabc2}
$$[\pa_{x_{m,l}^*},\pa_{x_{2n+1}}]=[\pa_{y_{m,l}^*},\pa_{x_{2n+1}}]=[\pa_{x_{m,l}^*},\pa_{y_{2n+1}}]=[\pa_{y_{m,l}^*},\pa_{y_{2n+1}}]=0.$$
\end{proposition}
\begin{proof} By  equations (\ref{bctlhierarchy}),(\ref{wevolution}),
(\ref{mevolution}),(\ref{btladdsym}) and (\ref{btllaxwaveaddsymm})
\begin{eqnarray*}
[\pa_{x_{m,l}^*},\pa_{x_{2n+1}}] {W}&=& \pa_{x_{m,l}^*}(({L}_1^{2n+1},0)_+{W})+\pa_{x_{2n+1}}((A_{1ml}({M}_1,{L}_1),0)_-{W})\\
&=&-[(A_{1ml}({M}_1,{L}_1),0)_-, ({L}_1^{2n+1},0)]_+{W}+[({L}_{1}^{2n+1},0)_+,(A_{1ml}({M}_1,{L}_1),0)]_-{W}\\
&&-({L}_1^{2n+1},0)_+(A_{1ml}({M}_1,{L}_1),0)_-{W}+(A_{1ml}({M}_1,{L}_1),0)_-({L}_1^{2n+1},0)_+{W}\\
&=& [({L}_1^{2n+1},0)_+,(A_{1ml}({M}_1,{L}_1),0)_- ]_+{W}+[({L}_1^{2n+1},0)_-,(A_{1ml}({M}_1,{L}_1),0)_- ]_+{W}\\
&&+[({L}_{1}^{2n+1},0)_+,(A_{1ml}({M}_1,{L}_1),0)]_-{W}
-[({L}_1^{2n+1},0)_+,(A_{1ml}({M}_1,{L}_1),0)_-]{W}\\
&=& [({L}_1^{2n+1},0)_+,(A_{1ml}({M}_1,{L}_1),0)_- ]{W}
-[({L}_1^{2n+1},0)_+,(A_{1ml}({M}_1,{L}_1),0)_-]{W}\\
&=&0.
\end{eqnarray*}
Note that, the term $[({L}_1^{2n+1},0)_-,(A_{1ml}({M}_1,{L}_1),0)_- ]_+$ in the third equality
and $[({L}_{1}^{2n+1},0)_+,$ \\$(A_{1ml}({M}_1,{L}_1),0)_+]_-$ vanish,
since $(P_-)_+=0$ and $(P_+)_-=0$ for $P\in\mathscr{D}$.

The proofs for others are almost the same.
\end{proof}
So now, the remaining work is to determine the explicit forms of $A_{iml}(i=1,2)$. Before
doing this,  a useful lemma is introduced first.
\begin{lemma}\label{btllemma}
\begin{eqnarray}
\Lambda^{-1}\varepsilon \Lambda=J^{-1}\varepsilon^T J,&&\Lambda\varepsilon^* \Lambda^{-1}=J^{-1}\varepsilon^{*T} J.\label{blemma}
\end{eqnarray}
\end{lemma}
\begin{proof}
By recalling $\varepsilon:=\textmd{diag}[s]\Lambda^{-1}$ and $J=((-1)^i\delta_{i+j,0})_{i,j\in\mathbb{Z}}=J^{-1}$, we compare the $(i,j)$-entries on each side of the first relation in (\ref{blemma}). For the left side,
\begin{eqnarray*}
(\Lambda^{-1}\varepsilon \Lambda)_{i,j}&=&(\Lambda^{-1}\textmd{diag}[s]\Lambda^{-1} \Lambda)_{i,j}\\
&=&(\textmd{diag}[s-1]\Lambda^{-1})_{i,j}=(i-1)\delta_{i-1,j},
\end{eqnarray*}
and on the right side
\begin{eqnarray*}
(J^{-1}\varepsilon^T J)_{i,j}&=&\sum_{\alpha,\beta}(J^{-1})_{i,\alpha} (\varepsilon^T)_{\alpha,\beta}J_{\beta,j}\\
&=&\sum_{\alpha,\beta}(-1)^i\delta_{i+\alpha,0} (\Lambda \textmd{diag}[s])_{\alpha,\beta}(-1)^\beta\delta_{\beta+j,0}\\
&=&(-1)^{i+j}( \textmd{diag}[s+1]\Lambda)_{-i,-j}\\
&=&(-1)^{i+j}(-i+1)\delta_{-i+1,-j}=(i-1)\delta_{i-1,j},
\end{eqnarray*}
where the facts $\Lambda^{T}=\Lambda^{-1}$ and
$(-1)^{i+j}\delta_{-i+1,-j}=(-1)^{(j+1)+j}\delta_{i-1,j}=-\delta_{i-1,j}$
are used. Thus the first relation in  (\ref{blemma}) can be got.

As for the second relation, we start from transposing the first relation in (\ref{blemma}), which leads to $\Lambda^{-1}\varepsilon^T \Lambda=J\varepsilon J^{-1}=-\varepsilon^*$, with remembering $\varepsilon^*:=J\varepsilon J^{-1}$ and $J^T=J^{-1}=J$. Therefore
\begin{eqnarray*}
J^{-1}\varepsilon^{*T} J&=&-J^{-1}(\Lambda^{-1}\varepsilon^T \Lambda)^{T} J\\
&=&-J^{-1}\Lambda^{-1}\varepsilon \Lambda J\\
&=&-J^{-1}(J^{-1}\varepsilon^T J) J\\
&=&-\varepsilon^T=-\Lambda\Lambda^{-1}\varepsilon^T \Lambda\Lambda^{-1}=\Lambda\varepsilon^* \Lambda^{-1}.
\end{eqnarray*}
\end{proof}
\noindent\textbf{Remark 1}: for the BTL hierarchy, by (\ref{bcwconstraints}) and (\ref{blemma}),
\begin{eqnarray}
{M}^T&=&({W}(\varepsilon,\varepsilon^*){W}^{-1})^T=({W}^{-1})^T(\varepsilon^T,\varepsilon^{*T}){W}^T\nonumber\\
&=&(J,J){W}(J^{-1},J^{-1})(\varepsilon^T,\varepsilon^{*T})(J,J){W}^{-1}(J^{-1},J^{-1})\nonumber\\
&=&(J,J){W}(J^{-1}\varepsilon^TJ,J^{-1}\varepsilon^{*T}J){W}^{-1}(J^{-1},J^{-1})\nonumber\\
&=&(J,J){W}(\Lambda^{-1}\varepsilon \Lambda,\Lambda\varepsilon^* \Lambda^{-1}){W}^{-1}(J^{-1},J^{-1})\nonumber\\
&=&(J,J){W}(\Lambda^{-1},\Lambda)(\varepsilon ,\varepsilon^*)(\Lambda, \Lambda^{-1}){W}^{-1}(J^{-1},J^{-1})\nonumber\\
&=&(J,J){W}(\Lambda^{-1},\Lambda){W}^{-1}{W}(\varepsilon ,\varepsilon^*){W}^{-1}{W}(\Lambda, \Lambda^{-1}){W}^{-1}(J^{-1},J^{-1})\nonumber\\
&=&(J,J)({L}_1^{-1},{L}_2^{-1})({M}_1 ,{M}_2)({L}_1,{L}_2)(J^{-1},J^{-1})\nonumber\\
&=&(J,J){L}^{-1}{M}{L}(J^{-1},J^{-1})\label{bmtranspose}.
\end{eqnarray}
Because of the constraints (\ref{bctlconstr}) on the Lax operators for
the BTL hierarchy, we can not just let
$A_{iml}={M}_i^m{L}_i^l$, which will lead to
contradiction. By considering the constraints (\ref{bctlconstr}), we have
the following proposition for $A_{iml}$. For convenience, denote
$A_{ml}(M,L)=(A_{1ml}(M_1,L_1),A_{2ml}(M_2,L_2))$.
\begin{proposition}\label{propaddisymmtodab3}
For the BTL hierarchy, it is sufficient to ask for
\begin{equation}\label{btladdsymmconstr}
A_{ml}({M},{L})^T=-(J,J)A_{ml}({M},{L})(J^{-1},J^{-1}),
\end{equation}
thus we can take
\begin{equation}\label{abtladdsymm}
A_{ml}({M},{L})={M}^m{L}^l-(-1)^l{L}^{l-1}{M}^m{L}.
\end{equation}
\end{proposition}
\begin{proof}
    Since the constraints (\ref{bcwconstraints}) for the BTL
hierarchy, the action of $\pa_{x_{ml}^*}$ and $\pa_{y_{ml}^*}$ on
the transpose of the wave operator ${W}$ can be obtained
in two different way. The first is to transpose the relation
(\ref{btladdsym}), that is
\begin{eqnarray*}
\pa_{x_{m,l}^*}{W}^T=-{W}^T(A_{1ml}({M}_1,{L}_1),0)^T_-,\quad
\pa_{y_{m,l}^*}{W}^T=-{W}^T(0,A_{2ml}({M}_2,{L}_2))^T_-.
\end{eqnarray*}
The second is to use the constraints (\ref{bcwconstraints}),
\begin{eqnarray*}
\pa_{x_{m,l}^*}{W}^T&=&\pa_{x_{m,l}^*}((J,J){W}^{-1}(J^{-1},J^{-1}))=-(J,J){W}^{-1}(\pa_{x_{m,l}^*}{W}){W}^{-1}(J^{-1},J^{-1})\\
&=&(J,J){W}^{-1}(A_{1ml}({M}_1,{L}_1),0)_-(J^{-1},J^{-1})\\
&=&{W}^{T}(J,J)(A_{1ml}({M}_1,{L}_1),0)_-(J^{-1},J^{-1}),\\
\pa_{y_{m,l}^*}{W}^T&=&\pa_{y_{m,l}^*}((J,J){W}^{-1}(J^{-1},J^{-1}))=-(J,J){W}^{-1}(\pa_{y_{m,l}^*}{W}){W}^{-1}(J^{-1},J^{-1})\\
&=&(J,J){W}^{-1}(0,A_{2ml}({M}_2,{L}_2))_-(J^{-1},J^{-1})\\
&=&{W}^{T}(J,J)(0,A_{2ml}({M}_2,{L}_2))_-(J^{-1},J^{-1}).
\end{eqnarray*}
The consistence of these two results lead to
\begin{eqnarray}
(A_{1ml}({M}_1,{L}_1),0)^T_-&=&(J,J)(A_{1ml}({M}_1,{L}_1),0)_-(J^{-1},J^{-1}),\nonumber\\
(0,A_{2ml}({M}_2,{L}_2))^T_-&=&(J,J)(0,A_{2ml}({M}_2,{L}_2))_-(J^{-1},J^{-1}),\label{btladdsymm1}
\end{eqnarray}
thus from the fact that if $A\in o(\infty)$, then $(A)_\pm\in o(\infty)$, it is sufficient to ask for
\begin{eqnarray}
(A_{1ml}({M}_1,{L}_1),0)^T&=&(J,J)(A_{1ml}({M}_1,{L}_1),0)(J^{-1},J^{-1}),\nonumber\\
(0,A_{2ml}({M}_2,{L}_2))^T&=&(J,J)(0,A_{2ml}({M}_2,{L}_2))(J^{-1},J^{-1}),\label{btladdsymm2}
\end{eqnarray}
that is,
\begin{equation}\label{btladdsymm3}
A_{ml}({M},{L})^T=-(J,J)A_{ml}({M},{L})(J^{-1},J^{-1}).
\end{equation}

From (\ref{bctlconstr}) and (\ref{bmtranspose}), we have
\begin{eqnarray*}
({M}^m{L}^l)^T&=&({L}^l)^T({M}^m)^T =(-1)^l(J,J){L}^l(J^{-1},J^{-1})(J,J){L}^{-1}{M}^m{L}(J^{-1},J^{-1})\\
&=&(-1)^l(J,J){L}^{l-1}{M}^m{L}(J^{-1},J^{-1}).
\end{eqnarray*}
Thus $({L}^{l-1}{M}^m{L})^T=(-1)^l(J,J){M}^m{L}^l(J^{-1},J^{-1})$, since $J^T=J^{-1}=J$. Therefore
$A_{ml}({M},{L})={M}^m{L}^l-(-1)^l{L}^{l-1}{M}^m{L}$ satisfies the requirement of (\ref{abtladdsymm}).
\end{proof}
Next, we will discuss the algebraic structure of the additional symmetry of the BTL hierarchy.
\begin{proposition}\label{propaddisymmtodab4}
Acting on the wave matrices ${W}$, $\pa_{x_{ml}^*}$ and $\pa_{y_{ml}^*}$ form a algebra of $w_\infty^B\times w_\infty^B$ in the case of the BTL hierarchy.
\end{proposition}
\begin{proof}
 According to (\ref{btladdsymmconstr}), we have
\begin{eqnarray}
&&[A_{ml}({M},{L}),A_{nk}({M},{L})]^T=[A_{nk}({M},{L})^T,A_{ml}({M},{L})^T]\nonumber\\
&=&(J,J)[A_{nk}({M},{L}),A_{ml}({M},{L})](J^{-1},J^{-1})=-(J,J)[A_{ml}({M},{L}),A_{nk}({M},{L})](J^{-1},J^{-1}),\nonumber
\end{eqnarray}
thus $A_{ml}({M},{L})$ constitute a closed algebra and we can set
\begin{equation}\label{wbalgebra}
   [A_{ml}({M},{L}),A_{nk}({M},{L})]:=C_{ml,nk}^{pq} A_{pq}({M},{L}).
\end{equation}
On the other hand, from (\ref{btladdsym})
\begin{eqnarray*}
&&[\pa_{x_{ml}^*},\pa_{x_{nk}^*}]{W}\\
&=&-\pa_{x_{ml}^*}((A_{1nk}({M}_1,{L}_1),0)_-{W})+\pa_{x_{nk}^*}((A_{1ml}({M}_1,{L}_1),0)_-{W})\\
&=&[(A_{1ml}({M}_1,{L}_1),0)_-,(A_{1nk}({M}_1,{L}_1),0)]_-{W}+(A_{1nk}({M}_1,{L}_1),0)_-(A_{1ml}({M}_1,{L}_1),0)_-{W}\\
&&-[(A_{1nk}({M}_1,{L}_1),0)_-,(A_{1ml}({M}_1,{L}_1),0)]_-{W}-(A_{1ml}({M}_1,{L}_1),0)_-(A_{1nk}({M}_1,{L}_1),0)_-{W}\\
&=&[(A_{1ml}({M}_1,{L}_1),0)_-,(A_{1nk}({M}_1,{L}_1),0)]_-{W}-[(A_{1ml}({M}_1,{L}_1),0)_-,(A_{1nk}({M}_1,{L}_1),0)_-]_-{W}\\
&&-[(A_{1nk}({M}_1,{L}_1),0)_-,(A_{1ml}({M}_1,{L}_1),0)]_-{W}\\
&=&[(A_{1ml}({M}_1,{L}_1),0)_-,(A_{1nk}({M}_1,{L}_1),0)_+]_-{W}+[(A_{1ml}({M}_1,{L}_1),0),(A_{1nk}({M}_1,{L}_1),0)_-]_-{W}\\
&=&[(A_{1ml}({M}_1,{L}_1),0),(A_{1nk}({M}_1,{L}_1),0)]_-{W},
\end{eqnarray*}
and similarly
$$[\pa_{y_{ml}^*},\pa_{y_{nk}^*}]{W}=[(0,A_{2ml}({M}_2,{L}_2)),(0,A_{2nk}({M}_2,{L}_2))]_-{W},$$
$$[\pa_{x_{ml}^*},\pa_{y_{nk}^*}]{W}=[(A_{1ml}({M}_1,{L}_1),0),(0,A_{2nk}({M}_2,{L}_2))]_-{W}.$$
Thus by (\ref{wbalgebra}), we get
\begin{eqnarray*}
[\pa_{x_{ml}^*},\pa_{x_{nk}^*}]{W}=C_{nk,ml}^{pq}\pa_{x_{pq}^*}{W},\quad [\pa_{y_{ml}^*},\pa_{y_{nk}^*}]{W}=C_{nk,ml}^{pq}\pa_{y_{pq}^*}{W},\quad[\pa_{x_{ml}^*},\pa_{y_{nk}^*}]{W}=0.
\end{eqnarray*}
That is,
\begin{equation}\label{wbalgebraexplicit}
   [\pa_{x_{ml}^*},\pa_{x_{nk}^*}]=C_{nk,ml}^{pq}\pa_{x_{pq}^*},\quad [\pa_{y_{ml}^*},\pa_{y_{nk}^*}]=C_{nk,ml}^{pq}\pa_{y_{pq}^*},\quad[\pa_{x_{ml}^*},\pa_{y_{nk}^*}]=0.
\end{equation}
\end{proof}
At last, let's introduce the generating function of the additional symmetry for the BTL hierarchy
\begin{equation}\label{generatoradditionalsymmetrybb}
     {Y}(\lambda,\mu)=({Y}_1(\lambda,\mu),{Y}_2(\lambda,\mu))=\sum_{m=0}^\infty \frac{(\mu-\lambda)^m}{m!}\sum_{l=-\infty}^\infty\lambda^{-l-m-1}(A_{m,m+l}({M},{L}))_-,
\end{equation}
and show its relation with the (adjoint) wave functions. For this, let's see a lemma developed in \cite{ASM95}.
\begin{lemma}\label{lemma}
Given two operators $U=(U_1,U_2)$, $V=(V_1,V_2)\in \mathscr{D}$ depending on $x$ and $y$, we have\footnote{$(A\otimes B)_{ij}=A_iB_j$}
\begin{eqnarray}
U_1V_1&=&res_{z}\frac{1}{z}\big(U_1\chi(z)\big)\otimes \big(V_1^T\chi^*(z)\big),\label{lemma1}\\
U_2V_2&=&res_{z}\frac{1}{z}\big(U_2\chi(z^{-1})\big)\otimes
\big(V_2^T\chi^*(z^{-1})\big),\label{lemma2}
\end{eqnarray}
where $res_z\sum_ia_iz^i=a_{-1}$.
\end{lemma}
\begin{proposition}\label{propaddisymmtodab5}
For the BTL hierarchy,
\begin{equation}\label{generatoradditionalsymmetryb}
     {Y}(\lambda,\mu)=\sum_{m=0}^\infty \frac{(\mu-\lambda)^m}{m!}\sum_{l=-\infty}^\infty\lambda^{-l-m-1}({M}^m{L}^{m+l}-(-1)^{m+l}{L}^{m+l-1}{M}^m{L})_-.
\end{equation}
We have
\begin{eqnarray}
{Y}_1(\lambda,\mu)&=&\lambda^{-1}({\Psi}_1(x,y;\mu)\otimes{\Psi}_1^*(x,y;\lambda)-{\Psi}_1(x,y;-\lambda)\otimes{\Psi}_1^*(x,y;-\mu))_-,\label{generatortodab1}\\
{Y}_2(\lambda,\mu)
&=&\lambda^{-1}({\Psi}_2(x,y;\mu^{-1})\otimes{\Psi}_2^*(x,y;\lambda^{-1})-{\Psi_2}(x,y;-\lambda^{-1})\otimes{\Psi}_2^*(x,y;-\mu^{-1}))_-.\label{generatortodab2}
\end{eqnarray}
\end{proposition}
\begin{proof}
Firstly, according to Lemma \ref{lemma},
\begin{eqnarray*}
{L}_1^{m+l-1}{M}_1^m{L}_1&=& {L_1}^{m+l-1}{M_1}^m{W}_1\Lambda {W}_1^{-1} \\
&=&res_{z}z^{-1}\big({L_1}^{m+l-1}{M_1}^m{W}_1\Lambda\chi(z)\big)\otimes\big(({W}_1^{-1})^T\chi^*(z)\big)\\
&=&res_{z}\pa_z^m\big(z^{m+l-1}{\Psi}_1(x,y;z)\big)\otimes{\Psi}_1^*(x,y;z),
\end{eqnarray*}
and similarly
\begin{eqnarray*}
{M}_1^m{L}_1^{m+l}=res_{z}\big(z^{-1+m+l}\pa_z^m{\Psi}_1(x,y;z)\big)\otimes{\Psi}_1^*(x,y;z),
\end{eqnarray*}
then
\begin{eqnarray*}
{Y}_1(\lambda,\mu)&=&res_{z}\sum_{m=0}^\infty \frac{(\mu-\lambda)^m}{m!}\sum_{l=-\infty}^\infty\lambda^{-l-m-1}\Big(\big(z^{-1+m+l}\pa_z^m{\Psi}_1(x,y;z)\big)\otimes{\Psi}_1^*(x,y;z)\\
&&-(-1)^{m+l}\pa_z^m\big(z^{m+l-1}{\Psi}_1(x,y;z)\big)\otimes{\Psi}_1^*(x,y;z)\Big)_-\\
&=&res_{z}\delta(\lambda,z)\Big(z^{-1}e^{(\mu-\lambda)\pa_z}{\Psi}_1(x,y;z)\otimes{\Psi}_1^*(x,y;z)\Big)_-\\
&&+res_{z}e^{(\mu-\lambda)\pa_z}\Big(\big(z^{-1}\delta(-\lambda,z){\Psi}_1(x,y;z)\big)\otimes{\Psi}_1^*(x,y;z)\Big)_-\\
&=&\Big(\lambda^{-1}e^{(\mu-\lambda)\pa_\lambda}{\Psi}_1(x,y;\lambda)\otimes{\Psi}_1^*(x,y;\lambda)\Big)_-\\
&&+res_{z}\Big(\big((\mu-\lambda+z)^{-1}\delta(-\lambda,z+\mu-\lambda){\Psi}_1(x,y;z+\mu-\lambda)\big)\otimes{\Psi}_1^*(x,y;z)\Big)_-\\
&=&\lambda^{-1}({\Psi}_1(x,y;\mu)\otimes{\Psi}_1^*(x,y;\lambda)-{\Psi}_1(x,y;-\lambda)\otimes{\Psi}_1^*(x,y;-\mu))_-,
\end{eqnarray*}
where $\delta(\lambda,z)=\sum_{n=-\infty}^\infty\frac{z^{n}}{\lambda^{n+1}}$ and $res_z(\delta(\lambda,z)f(z))=f(\lambda)$ is used.

Similarly,
\begin{eqnarray*}
&&{M}_2^m{L}_2^{m+l}-(-1)^{m+l}{L}_2^{m+l-1}{M}_2^m{L_2}\\
&=&res_z\Big(\big(z^{-1+m+l}\pa_z^m{\Psi}_2(x,y;z^{-1})\big)\otimes{\Psi}_2^*(x,y;z^{-1})\\
&&-(-1)^{m+l}\pa_z^m\big(z^{m+l-1}{\Psi}_2(x,y;z^{-1})\big)\otimes{\Psi}_2^*(x,y;z^{-1})\Big)
\end{eqnarray*}
and
\begin{eqnarray*}
{Y}_2(\lambda,\mu)
&=&\lambda^{-1}({\Psi}_2(x,y;\mu^{-1})\otimes{\Psi}_2^*(x,y;\lambda^{-1})-{\Psi_2}(x,y;-\lambda^{-1})\otimes{\Psi}_2^*(x,y;-\mu^{-1}))_-.
\end{eqnarray*}
\end{proof}
\section{The Additional Symmetry for the CTL Hierarchy}
In this section, the additional symmetry for the CTL hierarchy will be given.

Similar to the case of the above section, introduce additional independent variables
$x_{m,l}^*$ and $y_{m,l}^*$, and define the additional flows for the CTL hierarchy as follows
\begin{equation}\label{ctladdsym}
    \pa_{x_{m,l}^*}{W}:=-(A_{1ml}({M}_1,{L}_1),0)_-{W},\quad \pa_{y_{m,l}^*}{W}:=-(0,A_{2ml}({M}_2,{L}_2))_-{W}.
\end{equation}
The same computations lead to the following propositions.
\begin{proposition}\label{propaddisymmtodac1}
\begin{eqnarray}
\pa_{x_{m,l}^*}{\Psi}=-(A_{1ml}({M}_1,{L}_1),0)_-{\Psi},&& \pa_{y_{m,l}^*}{\Psi}=-(0,A_{2ml}({M}_2,{L}_2))_-{\Psi},\label{wctladdsym}\\
\pa_{x_{m,l}^*}{L}=[-(A_{1ml}({M}_1,{L}_1),0)_-,{L}],&& \pa_{y_{m,l}^*}{L}=[-(0,A_{2ml}({M}_2,{L}_2))_-,{L}],\label{lctladdsym}\\
\pa_{x_{m,l}^*}{M}=[-(A_{1ml}({M}_1,{L}_1),0)_-,{M}],&&
\pa_{y_{m,l}^*}{M}=[-(0,A_{2ml}({M}_2,{L}_2))_-,{M}].\label{mctladdsym}
\end{eqnarray}
\end{proposition}
\begin{proposition}\label{propaddisymmtodac2}
$$[\pa_{x_{m,l}^*},\pa_{x_{2n+1}}]=[\pa_{y_{m,l}^*},\pa_{x_{2n+1}}]=[\pa_{x_{m,l}^*},\pa_{y_{2n+1}}]=[\pa_{y_{m,l}^*},\pa_{y_{2n+1}}]=0.$$
\end{proposition}
Thus, (\ref{ctladdsym}) is indeed the symmetry of the CTL hierarchy. So now we will try to determine the explicit forms of $A_{iml}$. For this, let's see an important lemma first.
\begin{lemma}\label{clemma1}
\begin{eqnarray}
\varepsilon=K\varepsilon^T K^{-1},&&\varepsilon^* =K\varepsilon^{*T} K^{-1}.\label{clemma}
\end{eqnarray}
\end{lemma}
\begin{proof}
By noticing $K=\Lambda J$, (\ref{blemma}) leads to (\ref{clemma}).
\end{proof}
\noindent\textbf{Remark 2}:
In the case of the CTL hierarchy, according to (\ref{bcwconstraints})(\ref{clemma}) and also $K^T=K^{-1}=-K$,
\begin{eqnarray}
{M}^T&=&({W}(\varepsilon,\varepsilon^*){W}^{-1})^T=({W}^{-1})^T(\varepsilon^T,\varepsilon^{*T}){W}^T\nonumber\\
&=&(K^T,K^T){W}((K^{-1})^T,(K^{-1})^T)(\varepsilon^T,\varepsilon^{*T})(K,K){W}^{-1}(K^{-1},K^{-1})\nonumber\\
&=&-(K,K){W}(K,K)(\varepsilon^T,\varepsilon^{*T})(K,K){W}^{-1}(K^{-1},K^{-1})\nonumber\\
&=&(K,K){W}(K,K)(\varepsilon^T,\varepsilon^{*T})(K^{-1},K^{-1}){W}^{-1}(K^{-1},K^{-1})\nonumber\\
&=&(K,K){W}(\varepsilon,\varepsilon^{*}){W}^{-1}(K^{-1},K^{-1})\nonumber\\
&=&(K,K){M}(K^{-1},K^{-1}).\label{cmtranspose}
\end{eqnarray}

Just as the case of the BTL hierarchy, we can not just let $A_{iml}={M}_i^m{L}_i^l$ in the CTL hierarchy because of the constraints (\ref{bctlconstr}) on the Lax operators. Similar investigation leads to the following proposition.
\begin{proposition}\label{propaddisymmtodac3}
For the CTL hierarchy, it is enough to require
\begin{equation}\label{ctladdsymmconstr}
A_{ml}({M},{L})^T=-(K,K)A_{ml}({M},{L})(K^{-1},K^{-1}),
\end{equation}
so we can let
\begin{equation}\label{actladdsymm}
A_{ml}({M},{L})={M}^m{L}^l-(-1)^l{L}^{l}{M}^m.
\end{equation}
\end{proposition}
By the same way as the BTL hierarchy, we find
\begin{proposition}\label{propaddisymmtodac4}
Acting on the wave matrices ${W}$, for the CTL hierarchy, $\pa_{x_{ml}^*}$ and $\pa_{y_{ml}^*}$ form a algebra of $w_\infty^C\times w_\infty^C$.
\end{proposition}
At last, as for the generating function of the additional symmetry of the CTL hierarchy
\begin{equation}\label{generatoradditionalsymmetrycc}
     {Y}(\lambda,\mu)=({Y}_1(\lambda,\mu),{Y}_2(\lambda,\mu))=\sum_{m=0}^\infty \frac{(\mu-\lambda)^m}{m!}\sum_{l=-\infty}^\infty\lambda^{-l-m-1}(A_{m,m+l}({M},{L}))_-,
\end{equation}
we have
\begin{proposition}\label{propaddisymmtodac5}
For the CTL hierarchy,
\begin{equation}\label{generatoradditionalsymmetryc}
     {Y}(\lambda,\mu)=\sum_{m=0}^\infty \frac{(\mu-\lambda)^m}{m!}\sum_{l=-\infty}^\infty\lambda^{-l-m-1}({M}^m{L}^{m+l}-(-1)^{m+l}{L}^{m+l}{M}^m)_-.
\end{equation}
We have
\begin{eqnarray}
{Y}_1(\lambda,\mu)&=&(\lambda^{-1}{\Psi}_1(x,y;\mu)\otimes{\Psi}_1^*(x,y;\lambda)-\mu^{-1}{\Psi}_1(x,y;-\lambda)\otimes{\Psi}_1^*(x,y;-\mu))_-,\label{generatortodac1}\\
{Y}_2(\lambda,\mu)
&=&(\lambda^{-1}{\Psi}_2(x,y;\mu^{-1})\otimes{\Psi}_2^*(x,y;\lambda^{-1})-\mu^{-1}{\Psi}_2(x,y;-\lambda^{-1})\otimes{\Psi}_2^*(x,y;-\mu^{-1}))_-.\label{generatortodac2}
\end{eqnarray}
\end{proposition}
\begin{proof}
The same computations as the BTL hierarchy lead to
\begin{eqnarray*}
{M}_1^m{L}_1^{m+l}-(-1)^{m+l}{L}_1^{m+l}{M}_1^m
&=&res_z\Big(\big(z^{-1+m+l}\pa_z^m{\Psi}_1(x,y;z)\big)\otimes{\Psi}_1^*(x,y;z)\\
&&-(-1)^{m+l}z^{-1}\pa_z^m\big(z^{m+l}{\Psi}_1(x,y;z)\big)\otimes{\Psi}_1^*(x,y;z)\Big),\\
{M}_2^m{L}_2^{m+l}-(-1)^{m+l}{L}_2^{m+l}{M}_2^m
&=&res_z\Big(\big(z^{-1+m+l}\pa_z^m{\Psi}_2(x,y;z^{-1})\big)\otimes{\Psi}_2^*(x,y;z^{-1})\\
&&-(-1)^{m+l}z^{-1}\pa_z^m\big(z^{m+l}{\Psi}_2(x,y;z^{-1})\big)\otimes{\Psi}_2^*(x,y;z^{-1})\Big)
\end{eqnarray*}
and
\begin{eqnarray*}
{Y}_1(\lambda,\mu)&=&(\lambda^{-1}{\Psi}_1(x,y;\mu)\otimes{\Psi}_1^*(x,y;\lambda)-\mu^{-1}{\Psi}_1(x,y;-\lambda)\otimes{\Psi}_1^*(x,y;-\mu))_-,\\
{Y}_2(\lambda,\mu)
&=&(\lambda^{-1}{\Psi}_2(x,y;\mu^{-1})\otimes{\Psi}_2^*(x,y;\lambda^{-1})-\mu^{-1}{\Psi}_2(x,y;-\lambda^{-1})\otimes{\Psi}_2^*(x,y;-\mu^{-1}))_-.
\end{eqnarray*}
\end{proof}
\section{Conclusions and Discussions}
To summarize, we have constructed the additional symmetries for the BTL and CTL hierarchies in (\ref{btladdsym})(\ref{abtladdsymm}) and (\ref{ctladdsym})(\ref{actladdsymm}) respectively. And the additional symmetry flows on ${\Psi}$,
${L}$ and ${M}$ for the BTL and CTL hierarchies are given in Proposition \ref{propaddisymmtodab1}
and \ref{propaddisymmtodac1} respectively. When acting on the wave matrices ${W}$,
(a) $\pa_{x_{ml}^*}$ and $\pa_{y_{ml}^*}$ form a algebra of $w_\infty^B\times w_\infty^B$ in the case of the BTL hierarchy;
(b) For the CTL hierarchy, $\pa_{x_{ml}^*}$ and $\pa_{y_{ml}^*}$ form a algebra of
 $w_\infty^C\times w_\infty^C$. The generating functions of the additional symmetries for the
 BTL and CTL hierarchies are discussed in Proposition \ref{propaddisymmtodab5} and
 \ref{propaddisymmtodac5}, which may be helpful for the study of the ASvM formulas in the future.
 These results indicate again the essential difference between the BTL and CTL hierarchies from
 the point of view of the symmetry.

{\bf Acknowledgments}

{\noindent \small This work is supported by the NSF of China under
Grant No.10971109 \& 10971209, and K.C.Wong Magna Fund in Ningbo University.
Jingsong He is also supported by Program for NCET under Grant
No.NCET-08-0515.}


\begin{thebibliography}{99}
\bibitem{uenotaksasai}K. Ueno and K. Takasaki, Toda lattice hierarchy, In ``Group representations and systems of differential
equations" (Tokyo, 1982), 1-95, Adv. Stud. Pure Math., 4, North-Holland, Amsterdam, 1984.


\bibitem{toda}M. Toda, Vibration of a chain with nonlinear interaction. J. Phys. Soc. Jpn. 22(1967), 431-436.

\bibitem{DJKM}
E. Date, M. Jimbo, M. Kashiwara and T. Miwa, Transformation groups
for soliton equations, in {\it Nonlinear integrable systems -
classical theory and quantum theory\/}, ed. by M. Jimbo and T. Miwa,
 (World Scientific, Singapore, 1983)pp.39-119.

\bibitem{OS86}
A.Yu. Orlov and E.I. Schulman, Additional symmetries for integrable
systems and conformal algebra repesentation,  Lett. Math. Phys. 12
(1986), 171-179.

\bibitem{D93}
 L.A. Dickey, Additional symmetry of KP, Grassmannian, and the string equation, Mod.
Phys. Lett. A8(1993), 1259-1272.

\bibitem{ASM94}
M. Adler, T. Shiota and P. van Moerbeke, From the $w_
\infty$-algebra to its central extension: a $\tau$-function
approach. Phys. Lett. A 194 (1994), 33-43.

\bibitem{vM94}
P. van Moerbeke, Integrable fundations of string theory, in {\sl
Lectures on Integrable systems}, Edited by O. Babelon, P. Cartier,
Y. Kosmann-Schwarzbach(World Scientific,Singapore,1994) pp.163-267.

\bibitem{ASM95}
M. Adler, T. Shiota and P. van Moerbeke, A Lax representation for
the vertex operator and the central extension, Comm. Math. Phys. 171
(1995), 547-588.

\bibitem{D95}
L.A. Dickey, On additional symmetries of the KP hierarchy and Sato's
B\"acklund transformation, Comm. Math. Phys. 167 (1995), 227-233.
\bibitem{AvM1999} M. Adler and P. van Moerbeke, the spectrum of coupled random matrices,
Ann. Math. 149 (1999), 921-976.

\bibitem{takasaki1996}K. Takasaki, Toda Lattice Hierarchy and Generalized String Equations,
Comm.Math.Phys. 181(1996), 131-156.
\bibitem{fokas1981}A.S.Fokas and B.Fuchssteiner, The hierarchy of the Benjamin-Ono equation, Phys. Lett. A 86(1981), 341-345.

\bibitem{OF1982}W. Oevel, B. Fuchssteiner, Explicit fromulas for symmetries and conservation laws of
the Kadomtsev-Petviashvili equation. Phys. Lett. A 88(1982), 323-327.

\bibitem{F1983}B. Fuchssteiner, Mastersymmetries, Higher order time-dependent symmetries and con-
served densties of nonlinear evolution equation. Prog. Theor. Phys. 70(1983), 1508-1522.

\bibitem{Chen1983}H. H. Chen, Y. C. Lee, J. E. Lin, On a new hierarchy of symmetry for the
Kadomtsev-Petviashvili equation. Physica D 9(1983), 439-445.

\bibitem{fokas1987}P.M. Santini, M.J. Ablowitz and A.S. Fokas, On the initial value
problem for a class of nonlinear integral evolution equations
including the sine-Hilbert Equation, J. Math. Phys. 28(1987),
2310-2316.

\bibitem{fokas1988a}P.M. Santini and A.S. Fokas, Recursion operators and
bi-Hamiltonian structures in multidimensions I, Comm. Math. Phys.
115 (1988), 375-419.

\bibitem{fokas1988b}A.S. Fokas and P.M. Santini, Recursion operators and
bi-Hamiltonian structures in multidimensions II, Comm. Math. Phys. 116(1988), 449-474.





\bibitem{takasaki1993}K. Takasaki, Quasi-Classical Limit of BKP Hierarchy and W-Infinity Symmetries,
Lett.Math.Phys. 28(1993), 117-185.
\bibitem{L95}

J. van de Leur, The Adler-Shiota-van Moerbeke formula for the BKP
hierarchy, J. Math. Phys. 36 (1995), 4940-4951.

\bibitem{L96}
J. van de Leur, The $n$-th reduced BKP hierarchy, the string
equation and $BW_{1+\infty}$- constraints, Acta Appl.Math. 44(1996),
185-206.

\bibitem{Tu07}
M.H. Tu, On the BKP hierarchy: Additional symmetries, Fay identity
and Adler-Shiota- van Moerbeke formula, Lett. Math. Phys.
81(2007), 91-105.

\bibitem{Tu08}
H.F. Shen and M.H. Tu, On the string equation of the BKP
hierarchy, Int. J. Mod. Phys. A 24(2009), 4193-4208.


\bibitem{he2007}J. S. He, K. L. Tian, A. Foerster and W. X. Ma, Additional symmetries and string
equation of the CKP hierarchy, Lett.Math.Phys. 81(2007), 119-134.


\end{thebibliography}
\end{document}